\pdfoutput=1
\documentclass[11pt]{article}
\usepackage[margin=1in]{geometry}
\usepackage{amsmath,amssymb,amsthm}
\usepackage{graphicx}
\usepackage{booktabs}
\usepackage{hyperref}

\newtheorem{theorem}{Theorem}
\newtheorem{lemma}{Lemma}
\newtheorem{proposition}{Proposition}
\newtheorem{corollary}{Corollary}
\theoremstyle{definition}
\newtheorem{definition}{Definition}

\theoremstyle{remark}
\newtheorem{remark}{Remark}

\title{\textbf{Mesh Inference}\\[3pt]
\large A Formal Model of Collective Inference Without a Center}
\author{Hongwei Xu\\ SYM.BOT\\ \texttt{hongwei@sym.bot}}
\date{June 2026}

\begin{document}
\maketitle

\begin{abstract}
\noindent
We present a formal model of \emph{mesh inference}: how a population of independent agents, each
holding private state and exchanging only admitted, typed observations, derives a conclusion none of
them holds alone, with no central coordinator and no agent exposed. No agent shares weights, gradients,
or hidden state, and the agents may span different teams, networks, and organizations. Motivated by the
observation that asking a model is energy-minimizing inference, we model the mesh as a coupled free
energy that each agent relaxes locally. We show that a single admission/emission policy governs three
properties. First, mesh inference converges to a unique answer for any admission, symmetric or not,
because the coupling is always an M-matrix. Second, it is identification-complete: it derives the
centralized optimum exactly when the contributing views are carrier-connected. Third, it is
observation-only: no node transmits its internals, and confidentiality is the dual of identification---the
same rank condition that recovers collective knowledge bounds the recovery of private knowledge.
Content-addressed lineage is the only global side-channel. In the linear--Gaussian regime every derived
answer is determined, hence equal to the centralized optimum, at $O(\mathrm{diam}^2)$ latency, the
measured price of removing the center. One
such derivation is one turn of a center-free learning loop, which we formalize as architecture rather
than prove. The open problem we state is when asking improves the collective rather than corrupting it:
whether the non-linear closure derives an upgraded answer or a confident error. To our knowledge, this is the first formal
characterization of when a center-free, observation-only mesh recovers the centralized optimum.
\end{abstract}

\section{Introduction}
\label{sec:intro}
A growing body of agent protocols and ``internet of agents'' infrastructure lets autonomous agents
communicate. The harder question is whether a population can reason together: derive a conclusion that
the members jointly determine but none holds alone. We study this on a substrate where each agent runs
on its own device, owns its data, and never emits its weights, gradients, or hidden state. We call the
operation \emph{mesh inference}. It is not recall of a stored answer. It is the derivation of a
conclusion that no member could reach alone, composed from the members' partial views, computed with no
central coordinator and no agent exposed. The agents may belong to different teams, networks, and
organizations. The conclusion is real, since it lives in the union of the members' views, but reaching
it through purely local interaction is not automatic. This paper characterizes the conditions under
which it succeeds.

Mechanically, this reasoning is energy relaxation. Asking a model is, formally, energy-minimizing
inference: modern associative memories store patterns as the minima of an energy whose relaxation
dynamics is exactly transformer attention~\cite{hopfield}. We adopt the same view. A question clamps
coordinates of an energy, the dynamics relaxes, and the freed coordinates at equilibrium are the
answer. We write the primitive that drives one such relaxation $\mathrm{ask}(\cdot)$. The relaxation is
the inference, and its fixed point is the answer. The remaining question is which sovereign substrates
support it across a whole population. The answer is a condition on one object: the policy by which
agents admit and re-emit observations.

The full picture is a loop. An agent runs mesh inference and applies the result. It then learns from
what the application teaches and feeds that back, so the next inference, and the next agent with a
similar problem, draws on more. This paper proves one turn of that loop: the collective's answer is
real, determined, and computed without exposing any agent. It formalizes the loop that the turn sits
inside. And it identifies the open problem: whether, as the loop generalizes a learned answer to a new
problem, the derived answer can be trusted.

Standard distributed-intelligence designs place a center at exactly this point. Federated
learning~\cite{fedavg} keeps an aggregator. The dominant multi-agent paradigm is centralized training
with decentralized execution~\cite{maddpg}. Decentralized SGD removes the parameter server but still
optimizes a single shared model~\cite{dsgd}. Retrieval-augmented models keep a central
index~\cite{rag}. For an architecture whose value is per-node sovereignty~\cite{meshcognition}, that
center is the object to remove. Removing it at inference time, not just at execution time, is the bar.

That collective capability is the competitive layer of multi-agent systems is no longer speculative. In
2026 the largest vendors moved to own it. OpenAI's Frontier binds an enterprise's agents to a shared
knowledge layer over its CRM, ticketing, and data warehouses~\cite{frontier}. Microsoft's Agent
Framework supplies graph-based orchestration that routes a shared thread among agents~\cite{msagent}.
Teradata folds multimodal embeddings into a single governed vector store~\cite{teradatavs}. A further
class markets the property head-on: \emph{inference-mesh} and agentic-mesh orchestration
platforms---kamiwaza's Inference Mesh the on-the-nose example---place compute beside the data and
advertise answers no single node could produce, yet by their own account route each request to the
relevant nodes and pool the outputs at a synthesizer that merges and reconciles them~\cite{kamiwaza}.
Each of these systems works, and each keeps a center: the enterprise platforms because they live inside
one organization's trust boundary, where a coordinator holds the data and a single owner hosts the
index; the mesh-orchestration platforms because the synthesizer that routes and pools node outputs is
itself a center, and the point at which those outputs are aggregated in the clear, wherever the nodes
physically sit. In principle such systems perform
collective inference. What they do not do is perform it without a center. Mesh inference addresses the
regime that this boundary excludes: an \emph{internet of agents}~\cite{agentproto}, where inference
runs across teams, organizations, and networks, among parties that will not pool data, expose models,
or appoint a host. There a center is not available, and sovereignty becomes the enabling condition.

The question is no longer whether agents can be orchestrated, but whether a population that shares nothing
beyond admitted, typed observations can still answer collectively. This is not decentralization for its
own sake: in the regime we target the center is not a component one declines but a host no participant
will accept, so a center-free derivation is not the slower path to the answer but the only collective
inference the setting admits. And the property it must deliver---that the answer is the real collective
one, and that no member is exposed in reaching it---is a guarantee, which a deployment can assert but
only a formal model can establish: when it holds, how it fails, and the single policy that decides both.
That is why this operation must be proved, not merely built---and proving it is this paper.

\paragraph{Contributions.}
We show that a sovereign, observation-only mesh can infer what its population knows collectively: the
answer the members jointly determine but none holds alone. We characterize exactly when that inference
equals the centralized optimum. One such inference is one turn of a center-free learning loop, which we
also formalize. The contributions, tagged \textnormal{[proven]} / \textnormal{[framework]} /
\textnormal{[open]}:
\begin{enumerate}
\item \textnormal{[proven]} \emph{A formal model of mesh inference} (\S\ref{sec:model}), built from
first principles: the inference step $\mathrm{ask}(\cdot)$ is clamp-and-relax on a local free energy,
which is energy-based equilibrium inference made distributed and sovereign, with the
\emph{admission/emission policy} promoted to the single characterized object.
\item \textnormal{[proven]} \emph{One policy governs three properties}
(\S\ref{sec:a}--\ref{sec:unified}): one unconditional, and two that are a single
rank-and-connectivity condition read with opposite signs. Mesh inference is
\emph{convergent} unconditionally: the coupling is always an M-matrix, so it
terminates at a unique answer for any admission, symmetric or not. The other two are
the dual pair. It is \emph{identification-complete}: it derives the answer no member
holds, the centralized optimum, iff the policy satisfies provenance-aware admission
and non-selective source-novel forwarding, with carrier-graph connectivity the exact
criterion. And it is \emph{observation-only} by the same condition sign-flipped: no
node transmits weights, gradients, or raw state, and a private state stays
confidential exactly where the adversary's probes fail to span it---the
rank-and-connectivity that lets the collective recover its answer is what bounds what
an adversary recovers. All three reduce to one admission/emission policy, with
content-addressed lineage the single global side-channel, and the member's safety
falls out of the very condition that gives the collective its answer. \S\ref{sec:cost}
gives the honest cost, a compression--identification--rate frontier dialed by anchor
budget.
\item \textnormal{[framework]} \emph{The learning loop} (\S\ref{sec:loop}): the proven turn is the
inference half of a cycle (ask, apply, learn from the world, feed back) in which the collective closure
grows by empirical feedback. We formalize the dynamics as architecture, not theorems, bounded by the
principle that the mesh manufactures no information: it accumulates only what agents learn by acting.
\item \textnormal{[open]} \emph{The central open problem} (\S\ref{sec:scope}): \emph{when asking
improves the collective rather than corrupting it}. This is generalization correctness in the
non-linear closure, where an upgraded answer to a similar problem can be true or a confident error.
\end{enumerate}
The ingredients are classical and not ours. The relaxation that drives each step is energy-based
equilibrium inference in the sense of deep equilibrium models and equilibrium
propagation~\cite{bai,scellier,hopfield,densemem}; that a population minimizing local free energy
performs inference no single member can is the multi-agent free-energy
principle~\cite{friston,aifcc,flock,ecosystems}; and the construction further rests on predictive
coding~\cite{raoballard}, potential games~\cite{potential}, consensus, synchronization, and
distributed estimation~\cite{olfatisaber,kuramoto,acebron,consinnov}, graph signal
processing~\cite{gsp}. What is not classical is
the object itself. Communication between sovereign agents is handled, by interoperability
protocols~\cite{agentproto} and by federated retrieval~\cite{fdrag,cfedrag,fedragsurvey}. What has had
no formal model is the inference: how a sovereign mesh reasons, deriving what the population jointly
determines but no member holds, across organizations, with no center and no member exposed, and whether
that reasoning is sound. Mesh Cognition~\cite{meshcognition} names the architecture. This paper supplies
what it lacked, and what, to our knowledge, no prior work provides: a formal characterization of mesh inference
with mathematical guarantees (convergence, identification-completeness, and confidentiality) that reduce
all three to one characterized object, the admission/emission policy under which the mesh \emph{derives
what the population knows collectively but no member holds}, reliably and without exposing any member.

\section{Related Work}
\label{sec:related}

\paragraph{Equilibrium-based inference.}
The inference step $\mathrm{ask}(\cdot)$ is, formally, energy-based equilibrium inference: clamp a
subset of coordinates, relax to a fixed point, read the rest. This is the setting of deep
equilibrium models~\cite{bai}, which define an output as the implicit fixed point of a single layer,
and of equilibrium propagation~\cite{scellier}, whose inference phase clamps inputs and settles an
energy to its minimum, with the continuous Hopfield network~\cite{hopfield} as prototype and dense
associative memory~\cite{densemem} as its high-capacity generalization. These fix the energy and the
solver on a single device. Our contribution is to distribute the same operation across sovereign
nodes that never share state, and to characterize when the distributed relaxation still returns the
centralized fixed point.

\paragraph{Collective free-energy minimization.}
That a population of agents each minimizing a local free energy can perform inference no single agent
can is not new. It is the multi-agent reading of the free-energy principle~\cite{friston}. Because free
energy is extensive, agents minimizing their local free energy jointly minimize the
ensemble's~\cite{aifcc}, and recent work exhibits collectives holding information that exceeds any
member's internal model~\cite{flock}. We take this principle as given, not as a result. That literature
has begun to supply formal guarantees, such as posterior consistency and no-regret for
free-energy-minimizing agents~\cite{curiosity} and formal treatments of strategic multi-agent
inference~\cite{factorised}. But these are in the single-agent and game-theoretic settings, not the
distributed, observation-only regime we treat. Our characterization is of a different object: the
conditions on the admission/emission policy under which the local relaxation is convergent,
identification-complete, and observation-only, proved in the linear--Gaussian regime with explicit
wire-level constraints.

\paragraph{Collective intelligence by sharing belief.}
A parallel line makes collective intelligence depend on agents accessing one another's internal states.
Its programmatic statement is the active-inference ``ecosystems of intelligence'' vision~\cite{ecosystems}:
an ecology that attains shared intelligence by sharing beliefs through a common modeling language and
transaction protocol. Two models instantiate the mechanism: \cite{aifci} endows active-inference agents,
in stages, with Theory of Mind and goal alignment and watches system-level performance emerge in
simulation; \cite{tomteams} builds a network of Bayesian agents that infer teammates' beliefs from
observed communication and lifts a human team's performance above what its members reach unaided. In all
three, the collective gain is carried by internals crossing the node boundary---an agent modeling, or
receiving, another's beliefs. Our regime forbids exactly that. Mesh inference is observation-only
(\S\ref{sec:c}, Prop.~\ref{prop:moat}): no node represents, requests, or emits another's internals;
confidentiality is the dual of identification; and the collective optimum is reached by admitting and
re-emitting typed observations rather than by sharing belief. The Theory-of-Mind and belief-sharing
apparatus these works treat as the enabling capability is, in the linear--Gaussian regime, not required
for identification-completeness. They also differ in kind---\cite{ecosystems} a vision, \cite{aifci} a
simulation, \cite{tomteams} a human-subject study---and none states the admission/emission conditions
under which a distributed relaxation returns the centralized optimum, which is the object we characterize.

\paragraph{Federated retrieval.}
The applied analogue is federated retrieval-augmented generation, which answers queries over knowledge
fragmented across devices without sharing raw data~\cite{fdrag,cfedrag,fedragsurvey}. FD-RAG~\cite{fdrag}
targets multi-source queries over fragmented corpora, which is our setting, yet, like the rest of this
line, it retains a coordinator that aggregates retrieved context. Our $\mathrm{ask}(\cdot)$ has no such
center. The answer is the fixed point of a coordinator-free relaxation, and we prove the conditions
under which it equals the centralized result.

\paragraph{Local-as-view query answering.}
The identification side of $\mathrm{ask}(\cdot)$ has a precise formal neighbor in database theory:
\emph{local-as-view} (LAV) data integration, in which each source is described as a view over a global
schema and a query is answered by composing the views that cover it~\cite{lav,dataint}. Our
\S\ref{sec:b} is exactly this picture under a Gaussian metric. Each agent holds a lossy view
$y_i=C_i s^\star$, collective identifiability is view coverage ($\bigcup_i\mathrm{row}(C_i)$ full rank),
and the carrier graph is view reachability. This is the LAV pattern in a linear instantiation,
view-based least-squares recovery, rather than classical relational LAV (conjunctive-query rewriting,
certain answers under an open world). The structural correspondence is exact; the query algebra is not.
Classical LAV rewrites and routes the query at a mediator that holds the global view catalog, which is a
center, precisely what \S\ref{sec:intro} removes. We replace that catalog with content-addressed lineage
(Lemma~\ref{lem:rootid}), so a node can tell where a direction originates with no central registry. Its
inference-side cousin, Gaussian belief propagation, is exact in $O(\mathrm{diam})$ on tree-like carriers
and converges under walk-summability~\cite{walksum}, a condition implied by the diagonal dominance, and
indeed the M-matrix structure, that our convergence already establishes. On loopy carriers it converges
iteratively rather than in a single pass, at a rate set by the walk-summability margin, faster than
Jacobi diffusion but not $O(\mathrm{diam})$ in general. A lineage-routed message-passing realization of
$\mathrm{ask}$ thus inherits our admission/emission conditions unchanged, and we state it as the latency
program in \S\ref{sec:scope}.

\paragraph{Distributed associative memory.}
Closest in mechanism is distributed associative memory. \cite{distam} performs energy-based
associative retrieval through fully decentralized message passing, and \cite{fedhop} reconstructs
latent archetypes from partial, noisy, heterogeneous observations across clients, both our exact
retrieval form. The latter, however, ships each client's Hebbian operator to a central server for
aggregation. We remove the server, and beyond the retrieval rule we characterize the admission
policy that determines whether retrieval succeeds at all.

\paragraph{Agent interoperability protocols.}
A parallel effort standardizes how autonomous agents discover and message one another (MCP, A2A,
ACP, and the decentralized ANP~\cite{agentproto}), with recent work taking a semantic rather than
transport view of these exchanges~\cite{semproto}. These define the envelope, not the cognition: none
specifies an inference operator or a notion of collective identifiability. Our protocol~\cite{mmp} is the
semantic memory layer such transports would carry; the two are complementary, not competing.

\paragraph{A disambiguation.}
We distinguish our setting from ``decentralized AI inference'' in its now-dominant sense: markets
that distribute the \emph{compute} for an otherwise centralized model across rented or verifiable
hardware (Render, Akash, Gensyn, and related DePIN networks). That decentralizes execution; we
decentralize cognition, leaving each node's model, state, and data in place.

\paragraph{Foundations and contribution.}
The construction rests on classical results we use directly: distributed estimation and
consensus~\cite{consinnov,olfatisaber}, potential games~\cite{potential}, synchronization~\cite{kuramoto,acebron},
graph signal processing~\cite{gsp}, and predictive coding~\cite{raoballard}. Against all of the
above, our contribution is neither the energy-based inference step nor the collective-inference principle,
both of which predate us, but their synthesis under sovereignty together with its formal
characterization. To our knowledge, this is the first statement of the admission/emission policy
conditions under which center-free, observation-only mesh inference is convergent,
identification-complete, and equal to the centralized optimum.

\section{A formal model of mesh inference}
\label{sec:model}

We build the model from first principles, in four moves: the \emph{substrate}, meaning what a mesh is
and what sovereignty forbids; the \emph{policy}, the single object every guarantee will turn on; the
\emph{inference dynamics}, meaning how the mesh reasons, namely energy relaxation, of which
$\mathrm{ask}$ is one step; and the \emph{soundness program}, the three properties mesh inference must
satisfy to be a foundation one can build on: that it converges, that it derives what the population
jointly determines, and that it exposes no member.

\begin{definition}[Mesh]
A mesh, the substrate of the Mesh Memory Protocol~\cite{mmp}, is a graph $G=(V,E)$ with agents $V=\{1,\dots,N\}$ and neighborhoods $N_i$. Agent $i$ holds
private state $h_i\in\mathbb{R}^d$ and sovereign parameters $x_i$, and emits a public, lossy, typed
observation $o_i=g_{x_i}(h_i)\in\mathbb{R}^p$, $p<d$. Only $o_i$ is transmitted; $h_i,x_i$ never
leave $i$. Observations are typed over fields $F$ with orthogonal projectors $\{P_f\}_{f\in F}$.
Each agent carries a private target $\tau_i$ (its own data), which anchors it and prevents trivial
collapse.
\end{definition}

\begin{definition}[Admission/emission policy]
Agent $i$ admits a neighbor's observation through a receiver-autonomous, per-field gate
\[
A_i(o_j)=\sum_{f\in F}\alpha_{i,f}\,\beta\!\big(s_{i,f}(o_j)\big)\,m_j^f\,\big(P_f o_j-P_f\hat o_i\big),
\qquad \hat o_i=g_{x_i}(h_i),
\]
with sovereign weights $\alpha_{i,f}\ge0$, band-pass gate $\beta\in[0,1]$ (admit the mid-band:
neither redundant nor noise; this per-field band-pass admission is Symbolic-Vector Attention
Fusion~\cite{svaf}, whose mid-band preference follows the Wundt curve~\cite{berlyne}), relevance
score $s_{i,f}$, and \emph{emission carrier}
$m_j^f\in\{0,1\}$ indicating whether $j$ emits field $f$. The policy is the pair
$\mathcal{A}=(\alpha,\beta)$ together with the emission map $m$. Everything in the sequel is a
property of $\mathcal{A}$ and $m$.
\end{definition}

The carrier $m_j^f$ is the faithful expression of loss and sovereignty: $i$ can be influenced by
$j$ on field $f$ only if $j$ chose to emit $f$. This single term, an emission decision sitting
inside the receiver's update, drives both the convergence and the capability analysis.

\begin{definition}[Mesh inference]
Write the local free energy
\[
L_i=\tfrac12\sum_{j\in N_i}\sum_{f\in F} w_{ij}^f\,m_j^f\,\big\|P_f(o_i-o_j)\big\|^2
+\tfrac12\,\mu_i\|o_i-\tau_i\|^2
+\tfrac12\,\nu_i\|o_i-\mathsf{mem}_i(q)\|^2,
\]
where $w_{ij}^f=K_{ij}\,\alpha_{i,f}\,\beta(s_{i,f}(o_j))\ge0$ is the per-field admitted weight and $m_j^f$
the emission carrier, so the coupling term is the gradient of the admission gate $A_i$. Here $\tau_i$ is $i$'s private anchor, and $\mathsf{mem}_i(q)$ is a \emph{question-relevant}
read-out of $i$'s admitted-observation store (the memory term; kept separate from $\tau_i$ so the
private anchor is never conflated with admitted evidence). Each agent descends $L_i$ through its own
emission, treating admitted neighbor observations as fixed data. \emph{Mesh inference} is the
collective relaxation of this coupled free energy: to infer the answer to a question $q$, the mesh
clamps the question-field coordinates to $q$, lets every agent descend its own $L_i$ in concert, and
reads the freed coordinates at the joint equilibrium. We write one such step $\mathrm{ask}(q)$, the
inference primitive, and the answer is the fixed point the views jointly determine, not a value any
member stored.
\end{definition}

Because every result below holds for any nonnegative weights, the policy enters the guarantees
through its \emph{support} (which fields each agent admits, $w_{ij}^f>0$, and emits, $m_j^f=1$), not
through the band-pass magnitudes.

In the linear regime $\mathsf{mem}_i(q)$ is a similarity-weighted read-out of the store; replacing this single
quadratic term by a log-sum-exp associative energy over the full store is the entire step to the
attention-class regime of \S\ref{sec:scope}.

\section{Property (a): convergence}
\label{sec:a}

\subsection{Well-posedness of mesh inference}
Stack the configuration $o=(o_1,\dots,o_N)$. The global Hessian is
$H=L_w+\mathrm{diag}(\mu_i+\nu_i)$, where $L_w$ is the (per-field) admitted-weight graph Laplacian. On a
connected graph with some $\mu_i>0$, $H\succ0$.

\begin{proposition}[Every inference has a unique answer]
Partition coordinates into the clamped question set $Q$ and the free set $\bar Q$. The restricted
problem has Hessian $H_{\bar Q\bar Q}$, a principal submatrix of $H$, hence positive definite. So
$\mathrm{ask}(q)$ has a unique minimizer and terminates, for any $\nu_i\ge0$. Under forwarding the
effective operator is the asymmetric M-matrix of Lemma~\ref{lem:mmatrix} rather than $H$; uniqueness
persists because that operator is a nonsingular M-matrix on the anchored-reachable component, so the
clamped system retains a unique equilibrium.
\end{proposition}

\subsection{The reciprocity spectrum}
Whether the relaxation descends a common potential depends on the symmetry of the effective
coupling. Three regimes:
\begin{itemize}
\item \emph{Symmetric admission} ($w_{ij}=w_{ji}$): the dynamics is exactly the gradient flow of
the mesh free energy $\Phi(o)=\tfrac12\sum_{(i,j)\in E}w_{ij}\|o_i-o_j\|^2+\tfrac12\sum_i\mu_i\|o_i-\tau_i\|^2$
(a Tikhonov graph-signal denoising energy~\cite{gsp}),
which is then a Lyapunov function; on a connected graph $\Phi$ is strongly convex and the unique
equilibrium is reached exponentially at rate $\eta\,\lambda_{\min}(H)$.
\item \emph{Reversible admission} (a positive $\pi$ exists with $\pi_i w_{ij}=\pi_j w_{ji}$): a
$\pi$-weighted potential survives, extending the guarantee to a strictly larger class.
\item \emph{Asymmetric admission}: no scalar potential exists, and the quadratic certificate
$\|o-o^\star\|^2$ stops being monotone once the symmetric part of the coupling loses positive
definiteness (the anchor $\mathrm{diag}(\mu)$ sets that margin). Convergence itself is not lost. The
effective coupling is always a weakly diagonally dominant M-matrix (Lemma~\ref{lem:mmatrix}), so its
eigenvalues keep nonnegative real part for any nonnegative weights. Asymmetry costs the common potential
and the convergence rate, never stability.
\end{itemize}

\subsection{Convergence under forwarding: a structural argument}
Identification (\S\ref{sec:b}) will force agents to re-emit, or \emph{forward}, admitted directions
they did not natively observe, so that information crosses lossy intermediates. Forwarding is an
emission-side decision, and in the faithful receiver-autonomous model it is not symmetric: the
off-diagonals of the per-field operator are $-w_{ij}^f m_j^f$ versus $-w_{ji}^f m_i^f$, asymmetric
whenever $m_i^f\ne m_j^f$. The symmetric-potential argument therefore does not apply during propagation;
indeed $\lambda_{\min}$ of the operator's symmetric part goes negative at the propagation frontier.
Convergence nonetheless holds, by structure rather than by a common potential.

\begin{lemma}[Convergence under non-selective forwarding]
\label{lem:mmatrix}
Let agents forward \emph{non-selectively}: once $i$ admits a source-novel direction $f$ it emits $f$
thereafter, so emission sets grow monotonically $E_0\subseteq E_1\subseteq\cdots\subseteq E_T$.
Then $\mathrm{ask}(q)$ converges to the unique anchored equilibrium $o^\star$, exactly, with finite
transient. Anchoring sets the terminal rate through the effective diameter, never the stability.
\end{lemma}

\begin{proof}
The per-field operator $M^f$ is a Z-matrix: off-diagonals $M^f_{ij}=-w_{ij}^f m_j^f\le0$, and row
diagonal $\mu_i\mathbf{1}[\mathrm{obs}_i{=}f]+\sum_j w_{ij}^f m_j^f$ equal to the anchor term plus the
row's off-diagonal magnitude sum. Hence $M^f$ is weakly diagonally dominant, and on the
anchored-reachable component---where an anchored row makes it irreducibly diagonally dominant---a
nonsingular M-matrix, so every configuration is itself Hurwitz ($\mathrm{Re}(\lambda)>0$). Two
structural facts give convergence with no numerical bound. \emph{(i) Finite monotone switching:} each
carrier $m_i^f$ flips $0\to1$ exactly once, at the instant $i$ first admits a source-novel copy of $f$,
so the configuration is constant after a finite time $T$ ($\le\mathrm{diam}+1$ distinct configurations
along the propagation path), terminating in $E_T$, where every carrier-connected pair mutually emits
$f$ and the terminal operator is symmetric positive definite, hence Hurwitz. \emph{(ii) Finite-time segments:} each of the $\le\mathrm{diam}+1$ configurations is active for a
finite time, and a linear ODE with bounded operator carries a finite state to a finite state over a
finite interval. The trajectory on $[0,T]$ is therefore a finite concatenation of finite-time
linear-ODE segments, so the state at $T$ is finite---and this finiteness holds for \emph{any}
switching schedule, since the number of segments is bounded by $\mathrm{diam}+1$ regardless of when
the switches occur. Within a segment a non-normal $M^f$ can amplify transiently
($\sup_t\|e^{-M^f t}\|>1$, e.g.\ $\approx4.2$ at gain $4$), so the transient is \emph{finite per
instance}, not uniformly bounded; its magnitude may grow with diameter, its finiteness does not. From
$T$ onward the terminal symmetric positive-definite operator is Hurwitz and drives exponential
convergence to $o^\star$.
\end{proof}

\begin{remark}
The M-matrix form is cleaner than a dwell-time argument: with finitely many finite-time segments
there is no escape to rule out, so the \emph{finiteness} of the transient is schedule-independent
(its \emph{magnitude}, set by operator non-normality, may grow with diameter). Selective forwarding
($m_i^f\ne m_j^f$, including adversarial over-emission) drives the coupling's symmetric part
indefinite past the anchor margin, so the quadratic Lyapunov certificate and the common potential are
lost. But the operator remains an M-matrix with eigenvalues in the right half-plane, so it still
converges. Selective over-emission slows convergence, since it shrinks the spectral gap, and is
lineage-detectable (\S\ref{sec:unified}), but it does not destabilize. The one failure that is sharp
rather than gradual is loss of carrier connectivity (\S\ref{sec:b}).
\end{remark}

\subsection{Retrieval rate}
\begin{corollary}[Latency]
\label{cor:rate}
The slowest mode of $\mathrm{ask}$ is $\lambda_{\min}$ of the anchored terminal operator. For an
anchored path it is the Dirichlet--Laplacian ground mode, $\lambda_{\min}\approx(\pi^2/4)/N^2$
(empirically $N^2\lambda_{\min}\to2.4$, matching the path constant $\pi^2/4\approx2.47$),
\emph{independent of the anchor strength} $\mu$ once $\mu$ clears the vanishing threshold
$\sim\pi^2/4N^2$ that separates the weak- from the strong-anchor regime. Hence retrieval latency is
$O(\mathrm{diam}^2)$, the rate of graph heat-diffusion, set by geometry rather than anchor magnitude.
The effective diameter is the distance to the nearest anchor, so denser anchoring shortens it to
$O(r^2)$ in inter-anchor distance $r$. Validation therefore buys latency the same way it buys
capability: by adding validated anchors that shrink $r$, not by raising $\mu$, which above threshold
does nothing.
\end{corollary}

\subsection{Asynchronous realization}
The relaxation analyzed above is continuous-time. A deployed engine runs asynchronous, discrete
message passing, with no global clock and with delayed or dropped messages and node churn. Convergence
carries over, by the same M-matrix structure.

\begin{lemma}[Asynchronous convergence]
\label{lem:async}
Run the discrete iteration that updates each coordinate by the Jacobi/Gauss--Seidel map of $M$, with
each node activating at arbitrary times and reading possibly-stale neighbour values. Because
$M=L_w+\operatorname{diag}(\mu)$ is an irreducibly diagonally-dominant M-matrix, the Jacobi matrix
$B=I-D^{-1}M$ ($D=\operatorname{diag}M$) has spectral radius $\rho(B)<1$, so the iteration converges to
the continuous fixed point $o^\star=M^{-1}g$ for any bounded-delay schedule. If every node is anchored
($\mu_i>0$), $B$ is a strict max-norm contraction ($\|B\|_\infty<1$) and convergence is robust to
\emph{total} asynchronism (arbitrary delays~\cite{bertsekas}): no global clock or reliable delivery is required.
\end{lemma}

\begin{proof}
$M$ is a weakly diagonally-dominant Z-matrix (\S\ref{sec:a}), strictly dominant on anchored rows; on a
connected graph with an anchor it is irreducibly diagonally dominant, so $\rho(B)<1$ and both the
synchronous and bounded-delay asynchronous iterations converge to $o^\star$. Full anchoring makes every
row strictly dominant, so $\|B\|_\infty=\max_i\sum_{j\ne i}|M_{ij}|/M_{ii}<1$, a max-norm contraction,
to which the asynchronous-convergence theorem applies under arbitrary delays.
\end{proof}

\begin{remark}[Churn is transient disconnection]
A node leaving and rejoining is transient carrier disconnection, a live instance of property~(b). While
a node is absent its coordinate cannot track, leaving a residual that scales with churn rate and shrinks
with anchor density; forwarding re-establishes the carrier when it returns. Anchor density is therefore
the single dial for both retrieval latency (Corollary~\ref{cor:rate}) and asynchronous and churn
robustness, and anchors are validated observations (\S\ref{sec:b}), so validation buys both. We verify
this in simulation. The iteration reaches machine precision under $60\%$ message loss without churn
(loss is modeled as staleness until resend, hence bounded delay, which is benign), and the churn residual falls with anchor count, vanishing once the mesh is
strictly contractive.
\end{remark}

\section{Property (b): identification-completeness}
\label{sec:b}

We now ask when mesh inference realizes the answer a sovereign population determines collectively: one
that exists in the union of views but in no single agent, and which the centralized optimum names. Let
each agent hold a lossy view $y_i=C_i s^\star$ of a shared latent $s^\star$. The latent is
\emph{collectively identifiable} when $\bigcup_i\mathrm{row}(C_i)$ is full rank, though each $C_i$ alone
is rank-deficient. This is distributed estimation over a network~\cite{consinnov}. The contribution here is not the
estimator but the characterization of which admission/emission policies let purely local, sovereign
agents realize it---specifically the carrier-graph connectivity criterion (Thm~\ref{thm:ident}),
source-novel forwarding, and content-addressed lineage as the sole global side-channel, none of
which appear in classical distributed estimation, together with the observation-only dual
(\S\ref{sec:c}).

\begin{definition}[The mesh's collective inference]
\label{def:ci}
Let agent $i$'s \emph{view} be its constraint $V_i\!:\ y_i=C_i s^\star$ (knowledge of the projection
$C_i s^\star$), and let $\langle\cdot\rangle$ be the \emph{determination operator}: $\langle V\rangle$
is the set of queries whose answer is fixed by $V$, here the row span of its constraints. (The
inferential \emph{reasoning} reading, namely pattern completion and generalization, is the general abstraction
this realizes non-linearly in the associative regime, \S\ref{sec:scope}.) The collective determines
$\langle\bigcup_i V_i\rangle$, each member only $\langle V_i\rangle$. We formalize \emph{the mesh's}
collective inference as the closure gap
\[
  \mathcal N \;=\; \big\langle\textstyle\bigcup_i V_i\big\rangle\;\setminus\;\textstyle\bigcup_i\langle V_i\rangle,
\]
the answers fixed by the pooled views but by no member's own, even after it reasons to exhaustion. Here
$\mathrm{ask}(q)$ returns $q$'s answer in $\langle\bigcup_i V_i\rangle$, which is new knowledge precisely
when $q\in\mathcal N$. The gap is genuine. With $A$ knowing $s_1{+}s_2$ and $B$ knowing $s_1{-}s_2$,
neither determines $s_1$, yet the pair fixes $s_1=\tfrac12\big((s_1{+}s_2)+(s_1{-}s_2)\big)$, a fact no
member holds. The boundary is exact: $\mathcal N\subseteq\langle\bigcup_i V_i\rangle$, so this knowledge
is newly derivable, not new information. The collective realizes what its views jointly determine and
manufactures nothing beyond them. Collective identifiability ($\bigcup_i\mathrm{row}(C_i)$ full rank) is
exactly the case $\langle\bigcup_i V_i\rangle=$ the whole of $s^\star$-space.
\end{definition}

\subsection{Admission is local; identification is global}
\label{subsec:local}
A receiver scores incoming fields by novelty against its own anchors. But whether a locally redundant
view is globally rank-critical, the piece that completes the union, is invisible locally. Two
consequences follow, the second sharper than usually noticed.

\paragraph{Admission must preserve sources, not values.}
Redundancy in value space (the numbers look alike) is not redundancy in source space (same origin
versus independent observation). Independently sourced corroboration is what builds rank, so it must
not be gated. Content-addressed lineage distinguishes the two: with $\delta^{\mathrm{src}}_i(o_j)=
\mathbf{1}[\mathrm{roots}(\mathrm{lineage}(o_j))\not\subseteq R_i]$, the policy must admit if value-novel
or source-novel, and gate only if both are redundant.

\paragraph{Emission must carry, not merely admit.}
Hidden state never crosses the wire~\cite{mmp,meshcognition}, so an intermediate can recover a
rank-critical direction in its private state and never transmit it. The value then dies at an agent that
knows it: admission kept the edge open, but nothing carried the value onward. The fix is an
emission-side primitive, \emph{source-novel forwarding}: an agent re-emits an admitted source-novel
direction it did not natively observe.

\subsection{The carrier graph and the theorem}
\begin{definition}[Carrier graph]
For an identifying subspace $U$, the carrier graph $G_U$ has edge $\{i,j\}$ iff
$U\subseteq\mathrm{span}(\mathrm{emit}_i\cup\mathrm{emit}_j)$ (union masking: one emitting endpoint
suffices). When emissions are field-typed and $U$ is a single field $f$, $G_U=G_f$.
\end{definition}

\begin{theorem}[Identification-completeness]
\label{thm:ident}
Under provenance-aware admission and non-selective source-novel forwarding, $\mathrm{ask}(q)$
returns the centralized/Bayes-optimal answer---the conditional mean under the Gaussian model whose
prior the anchors $\tau_i,\mu$ encode and whose observation precisions the admitted weights
$w_{ij}^f$ carry---on the collectively-identifiable class \emph{iff}, for
every identifying subspace $U$, every observer of $U$ and every dependent receiver lie in a single
connected component of the carrier graph $G_U$. Forwarding is exactly the operation that extends
$G_U$ until the component connects; in the full-estimate case (every coordinate emitted) forwarding
is automatic and the condition holds trivially. Read through Def.~\ref{def:ci}, this is a
collective-intelligence statement: $\mathrm{ask}$ realizes the full collective closure
$\langle\bigcup_i V_i\rangle$, deriving every fact in $\mathcal N$, including those no member can, iff
the views each derivation composes are carrier-connected to the receiver that needs the answer.
Connectivity is what lets the reasoning route the fragments together; sever it and the fact falls back
to the prior.
\end{theorem}

\begin{proof}[Proof (field-typed case)]
By Lemma~\ref{lem:mmatrix} the relaxation runs the directed per-receiver carrier dynamics, but it
converges to the equilibrium of its terminal configuration, where every carrier-connected pair
mutually emits, which is exactly the union-masked Laplacian of the carrier graph $G_U$. So
$\mathrm{ask}$'s limit is the minimizer analyzed below, and we may reason on $G_U$ directly.
The objective separates per field; each field $f$ is an anchored graph-Laplacian solve on $G_f$,
anchored at the observers of $f$. The minimizer returns the true value at node $i$ iff $i$ is
connected in $G_f$ to an $f$-observer, and the prior otherwise. The general subspace case follows by
replacing coordinates with the projectors defining $U$. Necessity of forwarding is the witness of
\S\ref{subsec:local}: an intermediate that admits but does not emit leaves $G_U$ disconnected at itself despite an
open admission edge. Exactness of the criterion over native observers, single-crediting of
each source, and termination of forwarding are supplied by Lemma~\ref{lem:rootid}.
\end{proof}

\begin{lemma}[Root-identity, termination, and exactness of source-novel forwarding]
\label{lem:rootid}
Assume content-addressed lineage: a CMB's identity, hence its root set $\mathrm{roots}(\cdot)$, is
determined by its content and lineage. A native observation by agent $k$ mints a fresh root and
emits $\{\rho_k\}$; forwarding re-emits admitted content without incorporating new native data; a
remix mints exactly one fresh root per new native datum. Let $R_i$ be the roots agent $i$ has
admitted, $\mathrm{row}(C_\rho)$ the subspace observed at native root $\rho$, and
$\mathrm{reach}(i)$ the native roots whose observer shares $i$'s carrier-graph component. Recall
source-novelty: $o$ is source-novel to $i$ iff $\mathrm{roots}(o)\not\subseteq R_i$. Under
non-selective source-novel forwarding:
\begin{enumerate}
\item[\textnormal{(P1)}] \emph{Root-invariance:} $\mathrm{roots}(\mathrm{forward}(o))=\mathrm{roots}(o)$.
\item[\textnormal{(P2)}] \emph{Termination / anti-echo:} each agent admits each root at most once, so
each admission adds at least one new root to a set bounded by $|\mathcal{R}|$ and triggers at most one
forwarding event; forwarding halts in $\le|V|\cdot|\mathcal{R}|$ events, and a pure re-statement (no
new root) is never admitted or re-forwarded.
\item[\textnormal{(P3)}] \emph{Exactness:} the admitted-view subspace at $i$ equals
$\sum_{\rho\in\mathrm{reach}(i)}\mathrm{row}(C_\rho)$, each native source credited once, independent
of path multiplicity.
\end{enumerate}
Consequently the criterion of Theorem~\ref{thm:ident} is an exact ``iff'' over native observers:
$\mathrm{ask}(q)$ returns the centralized optimum iff every native root spanning the identifiable
subspace lies in $i$'s carrier-graph component.
\end{lemma}

\begin{proof}
\emph{(P1)} By content-addressing, $\mathrm{roots}(\cdot)$ is a function of content and lineage.
Forwarding re-emits admitted content with its lineage unchanged (no new native datum), so the
re-emission shares the content-address and $\mathrm{roots}(\mathrm{forward}(o))=\mathrm{roots}(o)$;
a remix unions exactly one fresh root. Thus a forwarded direction retains its root identity across
arbitrarily many hops.
\emph{(P2)} A message is admitted as source-novel only if it carries a root $\notin R_i$, and
admission adds those new roots to $R_i$. Since $R_i\subseteq\mathcal{R}$ grows monotonically and roots
are never removed, agent $i$ performs at most $|\mathcal{R}|$ admissions, each triggering at most one
forwarding event; the total is $\le|V|\cdot|\mathcal{R}|$, giving termination. A message whose roots
all lie in $R_i$, a pure re-statement or a re-forwarded copy of an already-held root, fails
source-novelty and is neither admitted nor re-forwarded, so bundling old roots with new ones cannot
re-trigger the old.
\emph{(P3)} By (P1) every forwarded copy of $\rho$ carries $\rho$ and lies in $\mathrm{row}(C_\rho)$;
by (P2) $i$ credits $\rho$ once. Path multiplicity therefore cannot inflate rank, since a forwarded copy
shares its source's subspace, while distinct native roots contribute additively, yielding the
stated subspace.
Finally $\mathrm{reach}(i)$ is exactly the native roots in $i$'s carrier-graph component, since
forwarding (root-preserving by (P1)) extends $G_U$ precisely along carried roots; so the
admitted-view rank equals $\mathrm{rank}\sum_{\rho\in\mathrm{reach}(i)}\mathrm{row}(C_\rho)$, which
matches the full-view rank iff every spanning native root lies in $i$'s component. The anchored
min-norm solve then returns the centralized optimum on precisely that subspace.
\end{proof}

\begin{remark}[The mechanism is content-addressing]
(P1) is not an extra assumption: content-addressed lineage already determines a CMB's identity from
its content and lineage, so re-emitting admitted content yields the same address and the same roots.
The provenance the protocol carries for audit is exactly what makes source-novelty well-defined and
forwarding terminating, which is lineage doing one more job. We verify this in simulation: on a
cyclic mesh, inherited-root forwarding terminates with each node's admitted roots equal to its
component's native roots, whereas a fresh-root-per-hop rule (violating (P1)) echoes without bound, forwarding events growing
geometrically per round (e.g.\ $2,4,8,16,\dots$ on a degree-2 ring); and forwarded copies of one source
leave subspace rank unchanged while independent sources add rank.
\end{remark}

\begin{remark}[Per-subspace, not per-field]
Forwarding must carry the identifying direction, not merely fields it touches: forwarding one
axis of a two-axis direction fails. Field-typed emissions (CAT7's orthogonal projectors) give the
clean per-field instantiation, but the criterion is subspace coverage of $G_U$.
\end{remark}

\begin{remark}[Reconciliation with the anti-echo rule]
Source-novel forwarding re-emits content for which the forwarder has no new domain data, which the
mesh protocol's anti-echo rule forbids as paraphrase~\cite{mmp}. The resolution is the value/source
distinction on the emission side. A forwarded direction carries an inherited lineage root that is new to
the receiver, not a freshly minted one (Lemma~\ref{lem:rootid}, P1), even with no new value, so it is
not paraphrase. Lineage licenses the forwarding the anti-echo rule would otherwise prohibit, and the
rule's guarantee is preserved because pure re-statements (no new root) stay forbidden. The forwarding
rule should additionally be non-selective, so that the carrier graph connects (identification,
\S\ref{sec:b}) and the convergence rate is preserved (Lemma~\ref{lem:mmatrix}), a constraint the
anti-echo rule does not currently express.
\end{remark}

\section{Property (c): observation-only}
\label{sec:c}

\begin{proposition}[Non-transmission]
\label{prop:moat}
Every term of $\mathrm{ask}$'s update is a function of agent $i$'s own emission, its private anchor
$\tau_i$, its admitted observations, and its sovereign weights. No neighbor's hidden state,
parameters, or gradients appear in any message. Hence $\mathrm{ask}$ returns centralized-equivalent
retrieval (Theorem~\ref{thm:ident}) while no node ever \emph{transmits} its weights, gradients, or raw
state.
\end{proposition}

Non-transmission is structural, true by construction of the wire format, but strictly weaker than
confidentiality: an adversary need not be sent a private quantity to reconstruct it from what is sent.
The honest question is whether the emission record identifies the private state, and the answer is
Property~(b) read with the opposite sign.

\begin{definition}[Probe-span]
\label{def:probe}
Write the influence of $i$'s private vector $\theta_i\in\mathbb{R}^d$ (the part of its data and state
that drives its emissions) as $o_i=C_i\theta_i$, lossy with $\operatorname{rank}C_i\le p<d$. An
adversary coalition induces a set $Q$ of queries and observes the \emph{converged answers}
$\mathrm{ask}$ returns: one rank-$\le p$ image $C_{i,q}\theta_i$ per query, not the relaxation
transient. The stacked record is $C_{\mathcal A}\theta_i$, $C_{\mathcal A}=[\,C_{i,q}\,]_{q\in Q}$, and
the adversary's \emph{probe-span} is $\mathcal P_{\mathcal A}=\operatorname{row}(C_{\mathcal A})
\subseteq\mathbb{R}^d$.
\end{definition}

\begin{proposition}[Confidentiality is identification with the opposite sign]
\label{prop:confid}
Against an adversary that observes the converged answers (Def.~\ref{def:probe}), the component of
$\theta_i$ in $\mathcal P_{\mathcal A}^{\perp}$ is non-identifiable: any $\theta_i'$ with
$\theta_i'-\theta_i\in\mathcal P_{\mathcal A}^{\perp}$ yields an identical record on $Q$, so no
estimator, at any computational budget, separates them. A single query hides a subspace of dimension
$\ge d-p$, and which directions stay hidden is a structural property of the realizable query family.
The result is a \emph{dichotomy}: either the induced probes span $\mathbb{R}^d$, and $\theta_i$ is
reconstructed exactly by generically $|Q|\gtrsim d/p$ queries (\emph{model inversion}), or the clamp
structure confines them to a proper subspace (\emph{probe-rank deficiency}), and the complementary
directions are confidential at any budget.
\end{proposition}

\begin{proof}
By Def.~\ref{def:probe} the record is exactly $C_{\mathcal A}\theta_i$, so it depends on $\theta_i$ only
through its image under $C_{\mathcal A}$; two private vectors with equal image are observationally
identical, and the indistinguishable set is the coset $\theta_i+\mathcal P_{\mathcal A}^{\perp}$ of
dimension $d-\dim\mathcal P_{\mathcal A}$. The dichotomy reads off $\operatorname{rank}C_{\mathcal A}$:
$\le p$ for one query; $=d$ once the induced family spans, which generic queries reach at
$|Q|\approx d/p$; and $<d$ permanently when the clamp structure leaves a fixed unprobed subspace, whose
complement no budget recovers. Additive emission noise leaves $\mathcal P_{\mathcal A}$ intact, since
least squares stays consistent on the spanned subspace, so only noise scaled to a query budget
(differential privacy) degrades reconstruction there. Scoping the adversary to the converged output is
what makes this exact. A warm start $o_i(0)=B\theta_i$ watched through the transient would add the mode
$e^{-Mt}o_i(0)$ and leak up to $\operatorname{rank}B$ beyond $C_{\mathcal A}$, but the query API exposes
only the settled answer, not the relaxation path, so that channel is closed by construction.
\end{proof}

\noindent The bound also covers a participating node that watches relaxation traffic, not only
an external querent reading answers: under cold start or warm start from a prior equilibrium,
$o_i(0)$ is either $\theta_i$-independent or a previous query's image already in
$\mathcal P_{\mathcal A}$, so the transient mode contributes nothing beyond $C_{\mathcal A}$. The
rank-$B$ leak needs a $\theta_i$-dependent high-rank initialization the construction never uses.

The \emph{rank} condition mirrors Theorem~\ref{thm:ident} exactly: identification needs the union of
views full-rank, and confidentiality is the same rank condition with the sign flipped, the private
subspace lying outside the adversary's accumulated span. The connectivity half is an analogy rather than
an identity, since the adversary's probe-span is fixed by the realizable query family, not by the
carrier graph; carrier disconnection, which defeats identification, only loosely mirrors the probe-rank
deficiency that shields private state. The reading is unsparing.
``Observation-only'' guarantees non-transmission (Prop.~\ref{prop:moat}) and per-query lossiness, not
unconditional secrecy. On the spanning branch of Prop.~\ref{prop:confid} an adaptive querent mounts a
model-inversion attack, while state and parameters that never enter any emission image stay hidden
unconditionally. The only defences for the queryable part are structural probe-rank deficiency, a
private subspace no admissible query can excite, or query-budgeted noise. The compression axis of
\S\ref{sec:cost} is precisely this leakage counted in bits.

This is the guarantee a hosted model structurally cannot match: no raw state, gradients, or weights
ever leave a node. We state it with its exact limit rather than as an unconditional moat.

\section{The unified characterization}
\label{sec:unified}
All three are governed by one object, the admission/emission policy, and together they make one thing:
a sovereign collective that derives what no member holds (Def.~\ref{def:ci}), reliably, and without
exposing any member. Convergence holds for free. The other two form a dual pair. Identification asks
when the collective recovers that answer; confidentiality asks when an adversary cannot recover a
private one. Both reduce to the same rank condition read with opposite signs (Fig.~\ref{fig:duality}):
\begin{itemize}
\item convergence is unconditional: the coupling is always an M-matrix (Lemma~\ref{lem:mmatrix}), so
it holds for any weights. Reciprocity and non-selective forwarding govern only whether a common
potential exists and how fast it converges, never whether it does;
\item identification-completeness needs source-preserving admission and source-novel forwarding with
carrier-graph connectivity (Theorem~\ref{thm:ident}). This is the collective-intelligence condition:
$\mathrm{ask}$ realizes the closure gap $\mathcal N$, deriving the answer no member holds, exactly where
the views that compose it are carrier-connected to the receiver;
\item observation-only is structural as \emph{non-transmission} (Prop.~\ref{prop:moat}) and, as
\emph{confidentiality}, the dual of identification (Prop.~\ref{prop:confid}). A private subspace stays
hidden exactly where the adversary's probes fail to span it, the same connectivity-and-rank condition
turned to the defender's favour.
\end{itemize}

\begin{figure}[h]\centering
\includegraphics[width=0.82\linewidth]{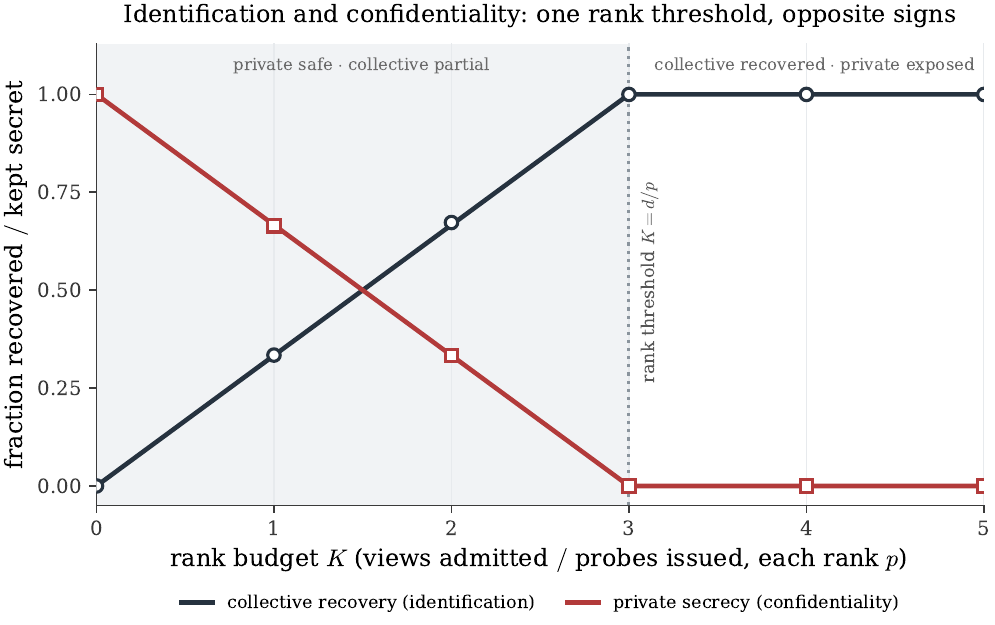}
\caption{\textbf{Identification and confidentiality are one rank threshold, read with opposite signs.}
A rank budget $K$ of rank-$p$ views ($d{=}6$, $p{=}2$) splits between the latent energy the mesh recovers
(identification, rising) and the private-direction energy an adversary cannot reconstruct
(confidentiality, falling); the two sum to one at every $K$ (markers measured, lines the rank law). At
$K{=}d/p$ the collective is fully identified and the private direction fully exposed---the condition that
completes collective recovery is the one that ends privacy.}
\label{fig:duality}
\end{figure}

\begin{table}[h]\centering\small
\begin{tabular}{@{}lll@{}}
\toprule
Property & Governing condition & Status \\
\midrule
Convergence & none: the coupling is always an M-matrix (Lem.~\ref{lem:mmatrix}) & unconditional\;[proven] \\[2pt]
Identification- & provenance-aware admission with non-selective source-novel & iff\;[proven] \\
completeness & forwarding; carrier-graph connectivity (Thm~\ref{thm:ident}) & \\[2pt]
Observation-only & non-transmission, structural (Prop.~\ref{prop:moat}); confidentiality & structural $+$ \\
 & $=$ probe-rank deficiency, the dual of identification (Prop.~\ref{prop:confid}) & conditional\;[proven] \\
\bottomrule
\end{tabular}
\caption{One policy, three properties. Convergence is free; identification and confidentiality are a
single rank-and-connectivity condition read with opposite signs.}
\label{tab:unified}
\end{table}

The global information that local agents cannot compute, source-independence and reciprocity, is
carried by \emph{content-addressed lineage}. It is load-bearing four times: constitutive for
source-preserving admission, constitutive for legitimate (source-novel) forwarding, the detector of
selective forwarding that would thin the carrier graph or slow the rate, and the carrier of the audit
trail. Lineage is constitutive for capability (admission and forwarding) and diagnostic for the rate
and connectivity it cannot itself enforce. These are distinct strengths, stated as such.

\noindent\textbf{The characterization.} \emph{Provenance-aware, source-preserving admission together
with non-selective source-novel forwarding makes $\mathrm{ask}$ identification-complete and
observation-only on the collectively-identifiable Gaussian class, with convergence holding
unconditionally (Lemma~\ref{lem:mmatrix}) and non-selectivity preserving its rate, and with lineage as
the single side-channel that lets purely local agents satisfy a set of globally-defined conditions.}
This is collective inference as a characterized operation. A sovereign population derives
what it jointly knows but no member holds, the answer the centralized optimum names, realized reliably
and without exposing any member, on the linear--Gaussian class. It is the inference half of a learning
loop (\S\ref{sec:loop}): one proven turn of a collective that also acts, learns, and feeds back.

\section{Experiments}
This section is mechanism verification: the does-the-theory-hold checks. The performance evaluation
(accuracy, latency, baselines) is \S\ref{sec:perf}. All experiments are synthetic and reproducible. The
accuracy experiment of \S\ref{sec:perf} is noisy linear--Gaussian (the centralized optimum is its noise
floor); the connectivity, forwarding, M-matrix, and latency experiments are exact (noiseless) algebraic
solves.

\paragraph{The intermediate that knows but does not say.}
On a chain $a\!-\!b\!-\!c$ where the rank-critical direction is observed only by $a$ and needed by
$c$, with $b$ not natively observing it: when $b$ admits the direction (recovering it exactly in its
own estimate) but does not re-emit it, $c$ recovers nothing; when $b$ forwards it, $c$ recovers it
exactly. The admission edge was open throughout, so the failure is emission, which isolates the
forwarding obligation cleanly. Carrier-graph connectivity predicts the linear solver exactly across
chains and a $3\times3$ grid, and the quantitative sweep over a forwarding probability $p$ shows
recovery exact while the carrier graph is connected, collapsing to the isolated level only as
connectivity is lost (Fig.~\ref{fig:fig1}).
\begin{figure}[h]\centering
\includegraphics[width=0.72\linewidth]{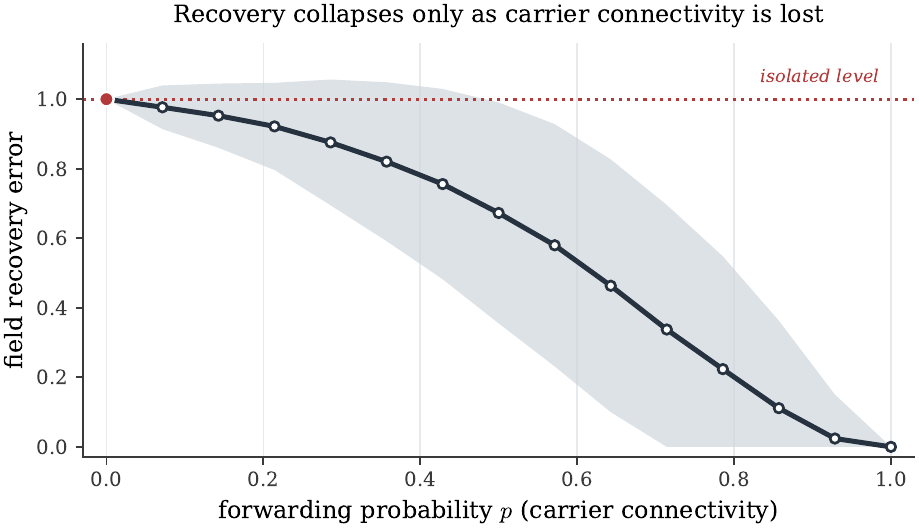}
\caption{Recovery is governed by carrier connectivity (logged; mean $\pm$ s.d.\ over $400$ runs). A
rank-critical field must cross intermediate agents; recovery error is exact while the carrier graph
stays connected ($p\!\to\!1$, non-selective forwarding) and rises to the isolated level only as
connectivity is lost ($p\!\to\!0$). Asymmetry alone, whether reciprocity-breaking or selective
over-emission, does not collapse recovery (Lemma~\ref{lem:mmatrix}); only disconnection does.}
\label{fig:fig1}
\end{figure}

\paragraph{Forwarding and convergence.}
Across $600$ frontier configurations (chains and grids, random partial-emission frontiers), every
carrier-masked operator is a diagonally dominant Z-matrix with $\mathrm{Re}(\lambda)\ge0$, and the
switched propagation converges to machine precision. Transient amplification is finite and
configuration-dependent---a gain-4 asymmetric carrier reaches $\sup_t\|e^{-M^f t}\|\approx4.2$---confirming
the finite transient of Lemma~\ref{lem:mmatrix} without entering its proof. Selective
over-emission does not break this. It drives the coupling's symmetric part indefinite past the anchor
margin (a two-node imbalance with gain $4$ gives $\lambda_{\min}(M_{\mathrm{sym}})=-0.30$), but $M$'s own
eigenvalues stay in the right half-plane (there, $5.41$ and $0.092$), so the operator remains an
M-matrix and converges. Over-emission only shrinks the spectral gap, slowing convergence (verified up to
gain $20$, where $\min\mathrm{Re}\,\lambda=+0.023$). Reciprocity-breaking is likewise benign for
accuracy: recovery error stays flat across the full asymmetry sweep, never approaching the isolated
level. The single mechanism that collapses recovery is carrier disconnection (Fig.~\ref{fig:fig1}).

\paragraph{The phase-encoded case.}
With phase emissions the relaxation is the Kuramoto model; coherence $r$ is negligible below the
mean-field critical coupling $K_c\approx1.60$ and rises sharply above it, making the coupling gain an
explicit design target: set it above the threshold and the mesh phase-locks.

\section{Performance: accuracy, latency, and leakage}
\label{sec:perf}
The framing of $\mathrm{ask}$ as equilibrium inference invites the questions an inference engine is
actually judged on: how accurate its answers are against a gold standard, how fast it returns them, and,
for a sovereign engine, how much they leak. We report all three on the regime we prove. Accuracy and
latency are measured against two baselines: the centralized optimum (one solver pooling every view, the
standard a sovereign mesh aims to match) and the isolated agent (no sharing, what it must beat).
Confidentiality is measured against an adaptive querent (\S\ref{sec:c}).

\emph{Accuracy.} Under carrier connectivity, $\mathrm{ask}$ realizes the collective answer, which the
centralized optimum names, with zero error (Theorem~\ref{thm:ident}). Consider a task no single agent
can solve: $N{=}20$ agents observe a shared latent through rank-deficient views (each rank $2\ll6$, the
union full rank), sharing only estimates, never views, models, or gradients. The noisy mesh reaches
mean error $0.041$, matching the centralized $0.041$ exactly, $20\times$ below the isolated $0.805$
(Fig.~\ref{fig:perf}a). Sovereignty costs nothing in accuracy: the forwarding fixed point is the
centralized optimum exactly (Thm.~\ref{thm:ident}); the cost is the $O(\mathrm{diam}^2)$ latency of Fig.~\ref{fig:perf}b. It also degrades gracefully:
as forwarding drops, accuracy falls to the isolated level exactly as carrier connectivity is lost, with
no intermediate cliff (Fig.~\ref{fig:fig1}). Disconnection reverts the unreachable coordinates to the
prior, so the answer is less complete, not pulled toward a wrong value.

\emph{Latency.} Latency is the price of removing the center, and we state it rather than hide it.
Retrieval is graph heat-diffusion: the slowest mode is the Dirichlet ground mode
$\lambda_{\min}\approx(\pi^2/4)/N^2$, so convergence takes $O(N^2)$ iterations in mesh diameter
(Fig.~\ref{fig:perf}b; measured log-log slope $1.9$), against a centralized solve's single round. The
rate is set by geometry, not anchor strength; denser validated anchoring shortens it to $O(r^2)$ in
distance-to-nearest-anchor (Corollary~\ref{cor:rate}), so more validated memory trades directly for
lower latency. This $O(\mathrm{diam}^2)$ is the cost of the diffusion solver (Jacobi relaxation), not a lower bound on the
inference: it reduces to $O(\mathrm{diam})$ under a lineage-routed message-passing realization---exact on
tree-like carriers, with loopy carriers inheriting the standard Gaussian belief-propagation
convergence caveat---whose mechanism, floor, and scope we state in \S\ref{sec:scope}.

\emph{Leakage.} A sovereign engine is also judged on what repeated queries reconstruct. Each ask
exposes a rank-$\le p$ projection of a node's private store (Prop.~\ref{prop:confid}), and an adaptive
querent accumulates them. Simulated against a node (Fig.~\ref{fig:perf}c), reconstruction error falls
from the single-query floor, a hidden subspace of dimension $\ge d-p$, as the querent's probes
accumulate, reaching machine precision once they span the store at $K\approx d/p$ queries (model
inversion). A probe-rank-deficient direction, never an answerable coordinate of any admissible query,
stays hidden at every budget. Leakage is thus the dual of identification, dialed by the query
budget (\S\ref{sec:cost}), not the all-or-nothing secrecy a one-line ``observation-only'' claim would
imply.

\emph{Scope.} This is performance on the proven linear--Gaussian turn. It is not language-model-grade,
and we do not present it as such. The non-convex associative regime, where the loop generalizes a
learned answer to a similar problem (\S\ref{sec:loop}), is where the open problem lives, and we name it
there rather than measure it here.

\begin{figure}[h]\centering
\includegraphics[width=\linewidth]{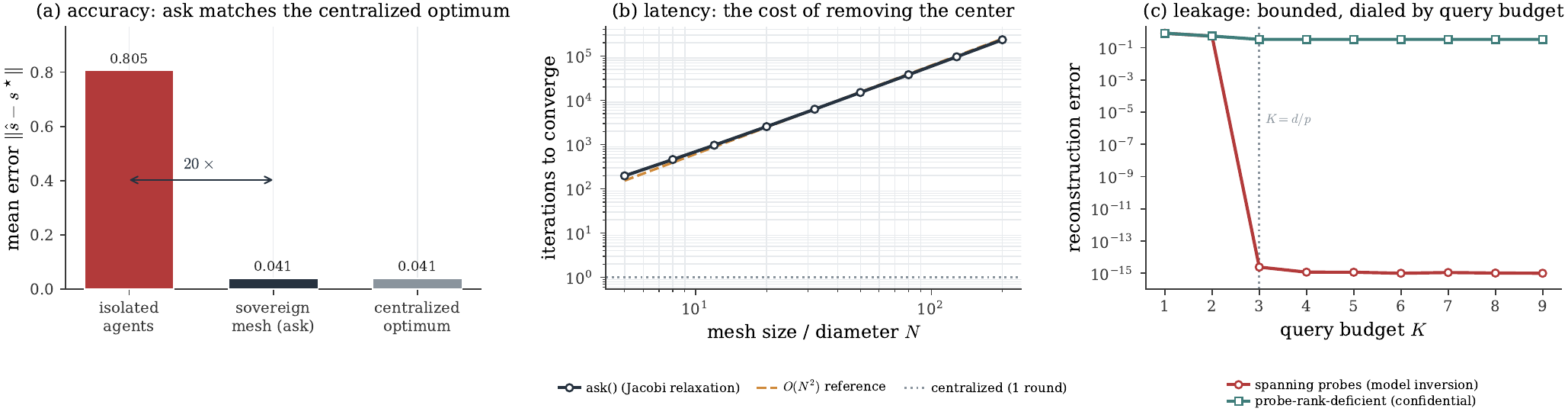}
\caption{Performance on a controlled linear--Gaussian inference task. (a) Accuracy: the sovereign mesh
($0.041$) matches the centralized optimum ($0.041$), $20\times$ below isolated agents ($0.805$); the
$N{=}20$ recovery task, sharing only estimates. (b) Latency: iterations for the Jacobi relaxation to
converge versus mesh size $N$ (anchored chain, diameter $N$), with the $O(N^2)$ reference and the
single-round centralized cost. The cost of removing the center is quadratic in the effective diameter
(the distance to the nearest anchor), independent of anchor strength. (c) Leakage: relative reconstruction error of a node's private vector ($d{=}6$) versus
query budget $K$, each ask leaking a rank-$p$ ($p{=}2$) projection. Spanning probes reconstruct it at
$K\!\approx\!d/p$ (model inversion); a probe-rank-deficient family, blind to a fixed private direction,
plateaus at the floor and never recovers it, confidentiality dialed by the budget.}
\label{fig:perf}
\end{figure}

\section{The cost of the construction: a compression--identification--rate frontier}
\label{sec:cost}
The construction has a genuine internal tension, and we state it rather than hide it. Non-selective
forwarding (needed for convergence) and source-novel forwarding (needed for identification) both grow
the emission dimension: an agent ends up emitting every source-novel direction it has admitted, pushing
$p$ toward $d$.

The two halves of Property~(c) (\S\ref{sec:c}) land on opposite sides of this frontier, and saying so
sharpens it. \emph{Non-transmission} (Prop.~\ref{prop:moat}) stands off the frontier: its proof never
invokes $p<d$, so it holds for any emission dimension, $p=d$ included, and sovereignty is invariant.
\emph{Confidentiality} (Prop.~\ref{prop:confid}) is the compression/privacy corner: a richer emission
($p\to d$) widens each per-query image $C_{i,q}$, so the adversary's probe-span covers more of the
private $\theta_i$ and the confidential subspace shrinks. What $p\to d$ trades is therefore
reconstruction-resistance, privacy counted in bits, and bandwidth, not sovereignty. The qualitative
guarantee (non-transmission) is invariant while the quantitative one (confidentiality) moves along the
frontier. The frontier is thus three-way, with non-transmission standing off it and confidentiality on
its privacy corner:
\begin{center}
\emph{compression/privacy} (small $p$, low leakage) \;vs\; \emph{identification} (needs forwarding for
carrier connectivity, grows $p$) \;vs\; \emph{convergence rate} (non-selective forwarding keeps the
terminal coupling symmetric, so a potential exists at the terminus; selective forwarding only slows, never destabilizes).
\end{center}
Anchor budget is the dial that places a deployment on this frontier: denser anchoring shortens
retrieval latency (Corollary~\ref{cor:rate}, via a smaller effective radius $r$), while the anchor
strength $\mu$ sets the margin tolerating selectivity (hence smaller $p$, less leakage) before the
convergence rate degrades (convergence itself always holds, Lemma~\ref{lem:mmatrix})---both at the
cost of more validated observations. The frontier is the honest
limit of the construction. It is also where the associative regime (\S\ref{sec:scope}) earns its place,
on capacity rather than latency: one-step associative retrieval escapes the diffusion latency that
compression aggravates only in the centralized setting, while the sovereign center-free reduction pays
the same $O(\mathrm{diam})$ floor (\S\ref{sec:scope}).

\section{The learning loop}
\label{sec:loop}
\noindent\emph{Status.}\enspace The results of this paper (\S\ref{sec:a}--\S\ref{sec:cost}) are
one characterized turn of the mesh, proven in the linear--Gaussian regime. This section gives the
operational model that turn sits inside. Its dynamics are architecture and open analysis, not theorems,
and every claim below is marked.

\medskip
A deployed mesh does not freeze. An agent asks, and the collective returns its current best answer
(proven). The agent applies that answer in the world. What the application teaches is new
information, earned empirically rather than produced by the mesh. The agent feeds it back as a fresh
observation, which admission, source-novel forwarding, and lineage (\S\ref{sec:b}--\S\ref{sec:c}, the
plumbing those sections build) carry to the others. The collective closure can only grow, so the next ask,
perhaps for a similar problem, draws on at least as much as the last.

\paragraph{The dynamics \textnormal{[framework]}.} Give each view a time index $V_i(t)$, and write
$V_i(t{+}1)=V_i(t)\cup\{\Delta_i(t)\}$, where $\Delta_i(t)$ is the constraint agent $i$ learned by
applying its answer, sourced from the world rather than the mesh. Then $\langle\bigcup_i V_i(t)\rangle$
is nondecreasing, and an ask at $t$ realizes the closure as it stands (Theorem~\ref{thm:ident}, turn by
turn). This is the two-phase shape of equilibrium propagation and active
inference~\cite{scellier,friston}. The ask is the inference phase, which reports what the collective
currently knows; the feedback is the learning phase, which grows what it knows. Because views only grow
and the space is finite-dimensional, the closure stabilizes in finitely many steps. The open question is not
whether the loop settles but whether what it settles on is correct.

\paragraph{Why the loop is honest, not magical \textnormal{[the boundary]}.} The turn manufactures no
information: $\mathcal N\subseteq\langle\bigcup_i V_i\rangle$, proven (\S\ref{sec:b}). The collective
grows only because agents act in the world and feed back what acting taught them. Information enters
through $\Delta_i(t)$, never from the mesh. The mesh is an accumulator of applied collective learning,
not an oracle. This is a stronger and more defensible claim than ``the collective knows more than its
parts,'' which stays false at every $t$.

\paragraph{Generalization is the open frontier \textnormal{[open]}.} A similar problem is the step the
linear turn does not cover. Matching a learned answer to a problem it was not learned on is the
non-linear $\langle\cdot\rangle$: the modern-Hopfield and attention regime, an associative energy with
exponential capacity and one-step retrieval~\cite{densemem,hopfield}. In our construction it is the
single substitution of the memory term $\tfrac12\nu_i\|o_i-\mathsf{mem}_i(q)\|^2$ by a log-sum-exp
associative energy over the admitted-observation store. There the collective can derive an upgraded
answer beyond any linear combination of views, and that answer can be a confident error: a generative
closure derives new knowledge that is true or false, and the loop has no guard that says which.
Determination stops being automatic, which is the paper's central open problem (\S\ref{sec:scope}).

\section{Scope and the open problem}
\label{sec:scope}
What is proved is the linear--Gaussian turn. Here $\mathrm{ask}$ is a convergent
conditional-mean engine that realizes the collective closure $\langle\bigcup_i V_i\rangle$, is
identification-complete up to collective identifiability, is observation-only, runs at
$O(\mathrm{diam}^2)$ latency, and makes every derived answer determined, hence equal to the centralized optimum. What is
framework is the loop of \S\ref{sec:loop}: the views move, the closure grows by applied feedback, and a
similar problem invokes the non-linear closure. The line between them is the line between this paper and
its program.

\paragraph{The central open problem: generalization correctness.} When the closure is non-linear, the
collective can derive an upgraded answer to a problem no member solved, and it can derive a wrong one
with the same confidence. The open problem is when asking improves the collective rather than corrupting
it: the conditions under which a generalized answer is guaranteed correct. This is the non-linear analog
of the carrier connectivity that makes every answer in the linear turn determined. There is a sharper,
dynamical facet. Coupled asynchronous inference and feedback could amplify an error: one fed back,
re-admitted, forwarded, and compounded. So correctness must hold not just per answer but under the
loop's own circulation.

\paragraph{Feasibility is not the obstacle.} The non-linear regime's distributed normalization is not an
open problem. The softmax is an associative, center-free reduction~\cite{onlinesoftmax,ringattn}, not
our result and not pursued here. What remains is the sovereign price on the two frontiers the proven results
already characterize. Emitting the associative regime's partial statistics is the model-inversion leak
of \S\ref{sec:c} (the secure-aggregation threat~\cite{secagg}), and a center-free reduction pays the
$O(\mathrm{diam})$ latency floor. On the linear side, the $O(\mathrm{diam}^2)$ diffusion solver already
reduces to that floor under a lineage-routed Gaussian message-passing realization, which is convergent
under the same M-matrix condition (Lemma~\ref{lem:mmatrix}) and exact on tree-like
carriers~\cite{lav,walksum}. These are the quantitative open items. Generalization correctness is the
qualitative one, and the one that decides whether the loop is collective intelligence or confident
noise.

\paragraph{Integrity under strategic agents.} The proven properties characterize a mesh whose agents
execute the admission/emission policy as specified. Convergence survives deviation---selective and
over-emitting carriers still settle (Lemma~\ref{lem:mmatrix})---but capability does not. A
self-interested participant can gate selectively to starve a rival's view, withhold a forward it owes,
or emit observations its own private state does not support, skewing or starving the collective answer.
The construction blunts the manipulation it can see: selective forwarding is lineage-detectable
(\S\ref{sec:unified}), over-emission only shrinks the spectral gap without destabilizing
(Lemma~\ref{lem:mmatrix}), and the symmetric regime descends a common potential (\S\ref{sec:a}), the
structure potential games~\cite{potential} exploit. But content-addressed lineage authenticates
provenance, not veracity, and detection is not deterrence: a common potential is not a proof that honest
emission is each agent's best reply. The open layer is mechanism design---conditions on the policy, or
on a lineage-backed penalty, under which truthful admission and non-selective forwarding are
incentive-compatible, an equilibrium no rational or adversarial agent profitably leaves. We supply the
detector; that detection suffices to keep the mesh honest is unproven.

\section{Conclusion}
Mesh inference can be made a well-posed operation. The mesh clamps the question, relaxes a local free
energy, and reads the answer. What comes back is what the population knows collectively: the answer its
members jointly determine but none holds alone, computed with no center and no member exposed. The
conditions under which a sovereign collective realizes that answer reduce to a single characterization
of the admission/emission policy (\S\ref{sec:unified}), with content-addressed lineage as the only
side-channel local agents need. The energy-based inference step and the collective-inference principle
both predate us; the contribution is their synthesis under sovereignty, named and proved. Within it,
identification and confidentiality are one fact read twice. The carrier connectivity and rank coverage
that let the collective derive what no member holds are exactly what bound an adversary's reconstruction
of what a member holds privately. So the condition under which the collective realizes its answer is the
condition that keeps each member's state safe.

That is one turn, and this paper proves it: in the linear--Gaussian regime every answer the collective
derives is determined, hence equal to the centralized optimum. But a deployed mesh does not freeze. The turn is the inference
half of a learning loop (\S\ref{sec:loop}): agents apply the collective's answer, learn from the world,
and feed back, so the closure grows and collective intelligence compounds. It compounds honestly,
because information enters by acting, never from the mesh. We formalize that loop as architecture and
prove only the turn. We state its frontier here rather than defer it: when does asking improve the
collective rather than corrupt it, and can a non-linear collective be trusted to derive an upgraded
answer to a similar problem rather than a confident error. The proven turn is the foundation that
question stands on. A collective that derives what no member holds, correctly and sovereignly, is the
thing a learning loop must not corrupt. Making it learn, provably, is the program this paper opens.

\end{document}